\documentclass{llncs}
\usepackage{amssymb}
\usepackage{amsmath}
\usepackage[ruled,vlined]{algorithm2e}
\usepackage{verbatim}
\newcommand{\classP}{\mathrm{P}}       
\newcommand{\NP}{\mathrm{NP}}  
\newcommand{\RP}{\mathrm{RP}}  
\newcommand{\totalP}{\mathrm{\bf TotalP}}
\newcommand{\TotalPP}{\mathrm{\bf TotalPP}}
\newcommand{\incremP}{\mathrm{\bf IncP}}
\newcommand{\incremPP}{\mathrm{\bf IncPP}}
\newcommand{\delayP}{\mathrm{\bf DelayP}}
\newcommand{\DelayPP}{\mathrm{\bf DelayPP}}
\newcommand{\enum}[1]{\textsc{Enum}\smash{\cdot}#1}

\title{Enumeration of the Monomials of a Polynomial and Related Complexity Classes}
\author{Yann Strozecki}
\institute{Universit\'e Paris 7 - Denis Diderot \\ \email{strozecki@logique.jussieu.fr}}

\begin{document}

\maketitle

\begin{abstract}
We study the problem of generating monomials of a polynomial in the context
of enumeration complexity. 
In this setting, the complexity measure is the delay between two solutions and the total time.
We present two new algorithms for restricted classes of polynomials, which have a good delay
and the same global running time as the classical ones. Moreover they are simple to describe,
use little evaluation points and one of them is parallelizable.

We introduce three new complexity classes, $\TotalPP$, $\incremPP$ and $\DelayPP$,
which are probabilistic counterparts of the most common classes for enumeration problems,
hoping that randomization will be a tool as strong for enumeration as it is for decision.
Our interpolation algorithms proves that a lot of interesting problems are in these classes like
the enumeration of the spanning hypertrees of a $3$-uniform hypergraph. 

Finally we give a method to interpolate a degree $2$ polynomials with an acceptable (incremental) delay. We also prove
that finding a specified monomial in a degree $2$ polynomial is hard unless $\mathrm{RP} = \NP$.
It suggests that there is no algorithm with a delay as good (polynomial) as the one we achieve for multilinear polynomials.

\end{abstract}

\section{Introduction }

Enumeration, the task of generating all solutions of a given problem, is an interesting generalization of decision and counting.
Since a problem typically has an exponential number of solutions, the way we study enumeration complexity is quite different
from decision. 
In particular, the delay between two solutions and the time taken by an algorithm relative to the number of solutions seem to be the most considered complexity measures.
In this paper, we revisit the famous problem of polynomial interpolation, that is to say finding the monomials of a polynomial from its values, with these measures in mind.

It has long been known that a finite number of evaluation points is enough to interpolate a polynomial and efficient procedures (both deterministic and probabilistic) have been studied by several authors \cite{ben1988deterministic,zippel1990interpolating,kaltofen2000early}. The complexity depends mostly 
on the number of monomials of the polynomial 
and on an a priori bound on this number which may be exponential in the number of variables.
The deterministic methods rely on prime numbers as evaluation points, with the drawback that they are very large.
The probabilistic methods crucially use the Schwarz-Zippel lemma, which is also a tool in this article, and efficient solving of particular linear systems.

As a consequence of a result about random efficient identity testing \cite{klivans2001randomness}, Klivans and Spielman 
give an interpolation algorithm, which happens to have an incremental delay.
In this vein, the present paper studies the problem of generating the monomials of a polynomial with the best possible delay.
In particular we consider natural classes of polynomials such as multilinear polynomials, for which we prove
that interpolation can be done efficiently.
Similar restrictions have been studied in other works about identity testing (the decision version of interpolation)
for a quantum model \cite{arvind2008quantum} or for depth $3$ circuits which thus define almost linear polynomials \cite{karnin2007black}.
Moreover, a lot of interesting polynomials are multilinear like the Determinant, the Pfaffian, the Permanent, the elementary symmetric polynomials or anything which may be defined by a syntactically multilinear arithmetic circuit.

In Sec. \ref{sec:inc} we present an algorithm which works for polynomials such that no two of their monomials use the same set of variables.
It is structured as in \cite{klivans2001randomness} but is simpler and has better delay, though polynomially related.
In Sec. \ref{sec:pdelay} we propose a second algorithm which works for multilinear polynomials; it has a delay polynomal in the numberof variables, which makes it exponentially better 
than the previous one and is also easily parallelizable.
In addition both algorithms enjoy a global complexity as good as the algorithms of the literature, are deterministic for monotone polynomials
and use only small evaluation points making them suitable to work over finite fields.

We describe in Sec. \ref{sec:classes} three complexity classes for enumeration, namely $\totalP$, $\incremP$, $\delayP$ which are now commonly used 
\cite{DBLP:journals/ipl/JohnsonP88,DBLP:journals/ipl/KavvadiasSS00,DBLP:journals/tocl/DurandG07,bagan-algorithmes} 
to formalize what is an efficiently enumerable problem. We introduce probabilistic variants of these classes,
which happen to characterize the enumeration complexity of the different interpolation algorithms. 
Their use on polynomials computable in polynomial time enable us to prove that well-known problems are
in these classes. Those problems already have better enumeration algorithms except 
 the last, enumeration of the spanning hypertrees of a $3$-uniform hypergraph, for which our method
gives the first efficient enumeration algorithm.

In the last section we discuss how to combine the two algorithms we have presented to interpolate degree $2$ polynomials with incremental delay. 
We also prove that the problem of finding a specified monomial in a degree $2$ polynomial is hard 
by encoding a restricted version of the hamiltonian path problem in a polynomial given by the Matrix-Tree theorem (see \cite{aigner2007course}).
Thus there is no polynomial delay interpolation algorithm for degree $2$ polynomials similar to the one for degree $1$
 because it would solve the later problem and would imply $\mathrm{RP} = \NP$.
Finally we compare our two algorithms with several classical ones and show that they are good with regard to parameters like
number of calls to the black box or size of the evaluation points.

Please note that most proofs are given in the appendix.
\section{Enumeration Problems}

In this section, we recall basic definitions about enumeration problems and
 complexity measures and we introduce the central problem of this article.

The computation model is a RAM machine as defined in \cite{bagan-algorithmes} which has,
in addition to the classical definition, an instruction $Write(A)$ which outputs the content of the register $A$. 
The result of a computation of a RAM machine is the sequence of integers which were in $A$
when the instructions $Write(A)$ were executed. 
For simplicity we consider that these integers encode words, and that the input
of the machine is also a word represented by suitable integers in the input registers. 
Let $M$ be such a machine and $x$ a word, we write $M(x)$ the result of the computation of $M$ on $x$.
The order in which the outputs are given does not matter, therefore $M(x)$ will denote 
the set of outputs as well as the sequence. We choose a RAM machine instead of a Turing machine since
it may be useful to deal with an exponential amount of memory in polynomial time, see for instance 
the enumeration of the maximal independent sets of a graph \cite{DBLP:journals/ipl/JohnsonP88}.

\begin{definition}[Enumeration Problem]
 Let $A$ be a polynomially balanced binary predicate, i.e. $A(x,y)\Rightarrow |y| \leq Q(|x|)$,
 for a certain polynomial $Q$. We write $A(x)$  for the set of $y$ such that $A(x,y)$.
 We say that a RAM machine $M$ solves the enumeration problem associated to $A$, $\enum{A}$ for short,
 if $M(x) = A(x)$ and there is no repetition of solutions in the computation.
\end{definition}

 Let $T(x,i)$ be the time taken by a machine $M$ to return $i$
outputs from the instance $x$. As for decision problems, we are interested by the total time
taken by $M$, namely $T(x,|M(x)|)$. We are also interested by the delay between two solutions, that is to say $T(x,i+1) - T(x,i) $.
 $M$ has an \emph{incremental delay} when it is polynomial in $|x|$ and $i$,
and $M$ has a \emph{polynomial delay} when it is polynomial in $|x|$ only.

A probabilistic RAM machine has a special instruction $rand$ which writes in a specific register 
the integer $0$ or $1$ with equal probability. All outcomes of the instruction $rand$ during a run of a RAM machine are independent.

\begin{definition}[Probabilistic enumeration]
We say that the probabilistic RAM machine $M$
solves $\enum{A}$ with probability $p$ if $P[A(x) = M(x)] > p$
and there is no repetition of solutions in the computation.
\end{definition}

We adapt the model to the case of a computation with an oracle, by a special instruction which calls the oracle on a word
 contained in a specific register and then writes the answer in another register in unit time.

In this article we interpret the famous problem of interpolating a polynomial given by a black box as a enumeration problem.
It means that we try to find all the monomials of a polynomial given by the number of 
its variables and an oracle which allows to evaluate the polynomial on any point in unit time.
This problem is denoted by $\enum{Poly}$ but will be solved in this article 
only on restricted classes of polynomials.

\section{Finding one Monomial}

In this section we introduce all the basic tools we need to build interpolation algorithms.
One consider polynomials with $n$ variables and rational coefficients.
A sequence of $n$ positive integers $\vec{e} =  (e_{1},\dots,e_{n}) $ characterizes the monomial $\vec{X}^{\vec{e}} = X_{1}^{e_{1}}X_{2}^{e_{2}} \dots X_{n}^{e_{n}}$.
We call $t$ the number of monomials of a polynomial $P$ written  $ P(\vec{X}) = \displaystyle{\sum_{1\leq j \leq t} \lambda_{j} \vec{X}^{\vec{e_{j}}}}$.

The degree of a monomial is the maximum of the degrees of its variables and the total degree 
is the sum of the degrees of its variables. 
Let $d$  (respectively $D$) denote the degree (respectively the total degree) of the polynomial we consider,
that is to say the maximum of its monomial's degree (respectively total degree).
In Sec. \ref{sec:pdelay} we assume that the polynomial is multilinear i.e. $d =1$ and $D$ is thus bounded by $n$.

We assume that the maximum of the bitsize of the coefficients appearing in a polynomial
is $O(n)$ to simplify the statement of some results, in the examples of Sec. \ref{sec:classes} it is even $O(1)$.
When analyzing the delay of an algorithm solving $\enum{Poly}$ we are interested in both 
the number of calls to the black box and the time spent between two generated monomials.
We are also interested in the size of the integers used in the calls to the oracle, 
since in real cases the complexity of the evaluation depends on it.

The \emph{support} of a monomial is the set of indices of variables which appears 
in the monomial. Let $L$ be a set of indices of variables, for instance a support,
 then $f_{L}$ is the homomorphism of
$\mathbb{Q}[X_{1},\dots,X_{n}]$ defined by 
$ \left\{ \begin{array}{c c}
     X_{i} \rightarrow X_{i} &\text{ if } i \in L \\
     X_{i} \rightarrow  0 &\text{otherwise } 
    \end{array} \right.
 $

From now on, we denote $f_{L}(P)$ by $P_{L}$. It is the polynomial obtained by substituting $0$ to every 
variable of index not in $L$, that is to say all the monomials of $P$ which have their support in $L$. 
We call $\vec{X}^{L}$ the multilinear term of support $L$, which is the product of all $X_{i}$
with $i$ in $L$.

\begin{lemma}
Let $P$ be a polynomial without constant term and whose monomials have distinct supports and $L$
a minimal set (for inclusion) of variables  such that $P_{L}$ is not identically zero.
Then there is an integer $\lambda$ such that $ P_{L} = \lambda \vec{X}^{L}$.
\label{basic}
\end{lemma}

The first problem we want to solve is to decide if a  polynomial given by a black box
 is the zero polynomial, a problem called  \emph{Polynomial Identity Testing}. We are especially interested in 
the corresponding search problem, i.e. giving explicitly one term and its coefficient. Indeed, we show in Sec. \ref{sec:inc} how to
turn any algorithm solving this problem into an incremental interpolation algorithm.

It is easy to see \cite{zippel1990interpolating} that a polynomial with $t$ monomials has to be evaluated in $t$ points to be sure 
that it is zero. If we do not have any a priori bound on $t$, then we must evaluate the polynomial on at least
$(d+1)^{n}$ $n$-tuples of integers to determine it. As we are not satisfied with this exponential complexity,
we introduce probabilistic algorithms, which nonetheless have a good and manageable bound on the error. 

\begin{lemma}[Schwarz-Zippel \cite{schwartz1980fast}]
\label{proba}
Let $P$ be a non zero polynomial with $n$ variables of total degree $D$, if  $x_{1}, \dots, x_{n}$ are randomly chosen 
in a set of integers $S$ of size $\frac{D}{\epsilon}$
then the probability that $P(x_{1}, \dots, x_{n}) = 0$ is bounded by $\epsilon$.
\end{lemma}

A classical probabilistic algorithm to decide if a polynomial $P$ is identically zero  can be derived from this lemma.
It picks $x_{1}, \dots, x_{n}$ randomly in $[\frac{D}{\epsilon}]$\footnote{We write $[x]$ for the set of integers between $1$ and $\lceil x\rceil$.} and calls the oracle to compute $P(x_{1}, \dots, x_{n})$.
If the result is zero, the algorithm decides that the polynomial is 
 zero otherwise it decides that it is non zero.
Remark that the algorithm never gives a false answer when the polynomial is zero.
The probability of error when the polynomial is non zero is bounded by $\epsilon$ thanks to Lemma \ref{proba}: 
\emph{Polynomial Identity Testing} is thus in the class recognizable by a polynomial time algorithm$\RP$.

This procedure makes exactly one call to the black box on points of size $\log(\frac{D}{\epsilon})$. 
The error rate may then be made exponentially smaller by increasing the size of the points.
There is an other way to achieve the same reduction of error. 
Repeat the previous algorithm $k$ times for $\epsilon = \frac{1}{2}$, that is to say the points are randomly
chosen in $[2D]$. If all runs return zero, then the algorithm decides that the polynomial is zero else it decides it is non zero. The probability
of error of this algorithm is bounded by $2^{-k}$, thus to achieve an error bound of $\epsilon$ we have to set $k =\log(\frac{1}{\epsilon})$. 
We denote by \emph{not\_zero}($P,\epsilon)$ the latter procedure, which is given as inputs a black box polynomial $P$ and the maximum probability of failure $\epsilon$. It uses slightly more random bits but it only involves numbers less than $2D$.

Up to Sec. \ref{sec:pdelay}, all polynomials 
have monomials with distinct supports and no constant term. This class of polynomials
contains the multilinear polynomials but is much bigger.
Moreover being without constant term is not restrictive since we can always
replace a polynomial by the same polynomial minus its constant term
that we compute beforehand by a single oracle call to $P(0,\dots,0)$.

We now give an algorithm which finds a monomial of a polynomial $P$,
in randomized polynomial time thanks to the previous lemmas. 
In this algorithm, $L$ is a set of indices of variables and $i$
an integer used to denote the index of the current variable.

\begin{algorithm}
\DontPrintSemicolon
 \KwData{A polynomial $P$ with $n$ variables and the error bound $\epsilon$}
\KwResult{A monomial of $P$}
\Begin{
$L \longleftarrow \{1, \dots, n\}$\;
\eIf{\emph{not\_zero}($P$,$\frac{\epsilon}{n+1})$}{
\For{$i = 1$ \KwTo $n$}{
\If{\emph{not\_zero}($P_{L \setminus\{i\}}$,$\frac{\epsilon}{n+1})$}{$L  \longleftarrow L \setminus \{i\}$}
}
\Return{\text{The monomial of support} $L$}
}
{\Return{``Zero''}}
}
  \caption{ find\_monomial}
\label{find_monomial}
\end{algorithm}

Once a set $L$ is found such that $P_{L}$ is a monomial $\lambda \vec{X}^{\vec{e}}$,
we must compute $\lambda$ and $\vec{e}$. The evaluation of $P_{L}$ on $(1,\dots,1)$ returns $ \lambda$. 
For each $i \in L$ the evaluation of $P_{L}$ on $X_{i} = 2$ and for $j \neq i$, $X_j=1$ returns $\lambda 2^{e_{i}}$.
From these $n$ calls to the black box, we compute $\vec{e}$ in linear time and thus output $\lambda \vec{X}^{\vec{e}}$.

We analyze this algorithm, assuming first that the procedure \emph{not\_zero}
never makes a mistake. We also assume that $P$ is not zero, which means that the algorithm
has not answered ``Zero''. 
In this case at the end of the algorithm, $P_{L}$ is not zero. In fact we remove an element 
from $L$ only if this condition is respected. As removing another element from $L$ would make
$P_{L}$ zero by construction, the set $L$ is minimal for the property of $P_{L}$ being non zero.
Then by Lemma  \ref{basic} we know that $P_{L}$ is a monomial of $P$,
which allows us to output it as previously explained.

Errors only appear in the procedure \emph{not\_zero}
with probability $ \frac{\epsilon}{n+1}$. Since we use this procedure $n+1$ times
we can bound the total probability of error by $\epsilon$.
The total complexity of this algorithm is $O(n\log(\frac{n}{\epsilon}))$ since each of the $n$ calls to the procedure 
 \emph{not\_zero} makes $O(\log(\frac{n}{\epsilon}))$ calls to the oracle in time $O(1)$.
We summarize the properties of this algorithm in the next proposition.

\begin{proposition}
\label{find}
Given a polynomial $P$ as a black box, whose monomials have distinct supports, Algorithm \ref{find_monomial} 
finds, with probability $1-\epsilon$, a monomial of $P$ by making $O(n\log(\frac{n}{\epsilon}))$ calls to the black box on entries of size $\log(2D)$.
\end{proposition}

\section{An Incremental Algorithm for Polynomials with Distinct Supports}\label{sec:inc}

%
We build an algorithm which enumerates the monomials of a polynomial
incrementally by using the procedure \emph{find\_monomial} defined in Proposition \ref{find}. 
Recall that incrementally means that the delay between two consecutive monomials 
is bounded by a polynomial in the number of already found monomials. 

We need a procedure \emph{subtract}($P$, $Q$) which acts as a black box
for the polynomial $P-Q$ when $P$ is given as a black box and $Q$ as an explicit
set of monomials with their coefficients. Let $D$ be the total degree of $Q$, $C$ a bound 
on the size of its coefficients and $i$ be the number of its monomials. 
One evaluates the polynomial \emph{subtract}($P$, $Q$) on points of size $m$ as follows:
\begin{enumerate}
 \item compute the value of each monomial of $Q$ in time $O(D\max(C, m))$
 \item add the values of the $i$ monomials in time $O(iD\max(C, m))$
 \item call the black box to compute $P$ on the same points and return this value minus the one we have computed for $Q$
\end{enumerate}

\begin{algorithm}
\label{incremental}
\DontPrintSemicolon
 \KwData{A polynomial $P$ with $n$ variables and the error bound $\epsilon$}
\KwResult{The set of monomials of $P$}
\Begin{
$Q \longleftarrow 0$\;
\While{\emph{not\_zero(subtract}($P$,$Q$),$\frac{\epsilon}{2^{n+1}})$}{
$M \longleftarrow$ find\_monomial(subtract($P$,$Q$),$\frac{\epsilon}{2^{n+1}})$\;
\textbf{Write}$(M)$\;
$Q \longleftarrow Q + M$ }
}
  \caption{ Incremental computation of the monomials of $P$}
\end{algorithm}

\begin{theorem}
Let $P$ be a polynomial whose monomials have distinct supports with $n$ variables, $t$ monomials and total degree $D$. Algorithm \ref{incremental} computes
the set of monomials of $P$ with probability $1- \epsilon$. The delay between the $i^{\text{th}}$ and $i+1^{\text{th}}$ outputted
monomials is bounded by $O( iDn^{2}(n +\log(\frac{1}{\epsilon})))$ in time and $O(n(n+\log(\frac{1}{\epsilon})))$
calls to the oracle. The algorithm performs $O(tn(n+\log(\frac{1}{\epsilon})))$ calls to the oracle on points of size $\log(2D)$. 
\label{PincremP}
\end{theorem}

\section{A Polynomial Delay Algorithm for Multilinear Polynomials }\label{sec:pdelay}

In this section we introduce an algorithm which enumerates the monomials of a multilinear polynomial with a polynomial delay.
This algorithm has the interesting property of being easily parallelizable, which is obviously not the case of the incremental one. 

Let $P$ be a multilinear polynomial with $n$ variables of total degree $D$.
Let  $L_{1}$ and $L_{2}$ be two disjoint sets of indices of variables and $l$ the cardinal of $L_{2}$. 
We can write $P_{L_{1}\cup L_{2}} = \displaystyle{ \vec{X}^{L_{2}} P_{1}(\vec{X}) + P_{2}(\vec{X})}$, 
where $\vec{X}^{L_{2}}$ does not divide $P_{2}(\vec{X})$.
We want to decide if there is a monomial of $P$, whose support contains $L_{2}$ and is contained in $L_{1} \cup L_{2}$,
which is equivalent to deciding wether  $P_{1}(\vec{X})$ is not the zero polynomial.
To do this, we define a univariate polynomial $H(Y)$ from $P_{L_{1}\cup L_{2}}$:
\begin{enumerate}
 \item substitute a randomly chosen value $x_{i}$ in $[2D]$ to $X_{i}$ for all $i \in L_{1}$
\item substitute the variable $Y$ to each $X_{i}$ with $i \in L_{2}$ 
\end{enumerate}
 The polynomial $H(Y)$ can be written $\displaystyle{ Y^{l} P_{1}(\vec{x}) + P_{2}(\vec{x},Y)}$. 
If $P_{1}$ is a non zero polynomial then $P_{1}(\vec{x})$ is a non zero constant with probability at least $\frac{1}{2}$ because of Lemma \ref{proba}.
Moreover $P_{2}(\vec{x},Y)$ is a polynomial of degree strictly less than $l$. Hence, to decide if the polynomial $P_{1}$
is not zero, we have to decide if $H(Y)$ is of degree $l$.

To this aim we do a univariate interpolation of $H(Y)$:
for this we need to make $l$ oracle calls on values from $1$ to $l$.
The time needed to do an interpolation thanks to these values, with $s$ a bound on the size of $H(i)$ for $1 \leq i \leq l$, is $O(l^{2}\log(s))$.
We improve the probability of error of the described procedure from $\frac{1}{2}$ to $\epsilon$ by repeating
it $\log(\frac{1}{\epsilon})$ times and name it \emph{not\_zero\_improved}($L_{1},L_{2},P,\epsilon$).

We now describe a binary tree which contains informations about the monomials of $P$.
The set of node of this tree is the pairs of list $(L_{1},L_{2})$ such that there exists
a monomial of support $L$ in $P$ with $L_{2} \subseteq L \subseteq L_{1}\cup L_{2}$.
Consider a node labeled by $(L_{1},L_{2})$, we note $i$ the smallest element of $ L_{1}$, it
 has for left child $(L_{1} \setminus \{i\}, L_{2})$ and for right child  $(L_{1} \setminus \{i\}, L_{2} \cup \{i \})$ if they exist.
The root of this tree is $([n],\emptyset)$ and the leaves are of the form $(\emptyset,L_{2})$.
There is a bijection between the leaves of this tree and the monomials of $P$:
 a leaf $(\emptyset,L_{2})$ represents the monomial of support $L_{2}$.

To enumerate the monomials of $P$, Algorithm \ref{pdelai} does a depth first search in this tree using \emph{not\_zero\_improved}
and when it visits a leaf, it outputs the corresponding monomial thanks to the procedure \emph{coefficient}($P$, $L$) that we now describe.
We have $L$ of cardinality $l$ the support of a term and we want to find its coefficient. 
Consider $H(Y)$ built from $L_{1} = \emptyset$ and $L_{2} = L$, 
the coefficient of $Y^{l}$ in this polynomial is the coefficient of the 
monomial of support $L$.
We interpolate $H(Y)$ with $l$ calls to the oracle as before and return this coefficient.

\begin{algorithm}[h]
\DontPrintSemicolon

 \KwData{A multilinear polynomial $P$ with $n$ variables and the error bound $\epsilon$}
\KwResult{All monomials of $P$}
 \Begin{
 Monomial$(L_{1},L_{2},i) =$ \;
\eIf{ $i = n+1$}
{\textbf{Write}(coefficient($P,L_{2}$))}
{\If{ \emph{not\_zero\_improved}$( L_{1}\setminus \{i\}, L_{2}, P, \frac{\epsilon}{2^{n}n})$}
    {Monomial$(L_{1}\setminus \{i\},L_{2},i+1)$}
\If{\emph{not\_zero\_improved}$( L_{1}\setminus \{i\} , L_{2}\cup \{i\}, P, \frac{\epsilon}{2^{n}n})$}
{Monomial$(L_{1}\setminus \{i\},L_{2}\cup \{i\},i+1)$}
}
in Monomial$([n],\varnothing,1)$ \;
}
  \caption{A depth first search of the support of monomials of $P$, recursively written}
  \label{pdelai}
\end{algorithm}

\begin{theorem}
Let $P$ be a multilinear polynomial with $n$ variables, $t$ monomials and total degree $D$. 
Algorithm \ref{pdelai} computes the set of monomials of $P$ with probability $1- \epsilon$. The delay between the $i^{\text{th}}$ and $i+1^{\text{th}}$ outputted
monomials is bounded in time by $O(D^{2}n^{2}\log(n)(n + \log(\frac{1}{\epsilon})))$ and by $O(nD(n + \log(\frac{1}{\epsilon})))$ oracle calls. 
The whole algorithm performs $O(tnD(n + \log(\frac{1}{\epsilon})))$ calls to the oracle on points of size  $O(\log(D))$.
\label{PdelayP}
\end{theorem}

There is a possible trade-off in the way  \emph{not\_zero\_improved} and \emph{coefficient} are implemented: 
if one knows a bound on the size of the coefficients of the polynomial and use exponentially bigger evaluations points then
one needs only one oracle call. The number of calls in the algorithm is then less than $tn$
which is close to the optimal $2t$.

Remark that when a polynomial is monotone (coefficients all positive or all negative) and is evaluated on positive points, 
the result is zero if and only if it is the zero polynomial. Algorithms \ref{incremental} and \ref{pdelai}
may then be modified to work deterministically for monotone polynomials with an even better complexity. 

Moreover both algorithms work for polynomials over $\mathbb{Q}$ but we can extend them to work over finite fields.
Since they only use evaluation points less than $2D$, polynomials over any field of size more than $2D$
can be interpolated with very few modifications, which is good in comparison with other classical algorithms.

\section{Complexity Classes for Enumeration}\label{sec:classes}

In this part the results about interpolation in the black box formalism are transposed into 
more classical complexity results. 
We are interested in enumeration problems defined by predicates $A(x,y)$
such that there is for each $x$ a polynomial $P_{x}$ whose monomials are in bijection with $A(x)$. 
If $P_{x}$ is efficiently computable, an interpolation algorithm gives an effective way of enumerating its monomials 
and thus to solve $\enum{A}$. 

\begin{example}
\label{det}
 We associate to each graph $G$ the determinant of its adjacency matrix. The monomials of this multilinear polynomial
are in bijection with the cycle covers of $G$. Hence the problem of enumerating the monomials of $\det(M)$ is equivalent to enumerating
the cycle covers of $G$.
\end{example}

The specialization of different interpolation algorithms to efficiently computable polynomials
naturally correspond to three ``classical'' complexity classes for enumeration 
and their probabilistic counterparts.
We present several problems related to a polynomial as in Example \ref{det} to 
illustrate how easily the interpolation methods described in this article
produce enumeration algorithms for combinatorial problems. Although the first two examples
already had efficient enumeration algorithms, the last did not, which shows that interpolation methods
 can bring new results in enumeration complexity.

In all the following definitions, we assume that the predicate which defines
the enumeration problem is decidable in polynomial time, that is to say the corresponding decision problem 
is in $\classP$. 

\begin{definition}
A problem $\enum{A}$ is decidable in polynomial total time $\totalP$ (resp. probabilistic polynomial total time $\TotalPP$)
if there is a polynomial $Q(x,y)$ and a machine $M$ which solves $\enum{A}$ (resp. with probability greater than $\frac{2}{3}$)
 and satisfies for all $x$, $T(x,|M(x)|) < Q(|x|,|M(x)|)$.
\end{definition}

$\TotalPP$ is very similar to the class $\mathrm{\bf{BPP}}$ for decision problems.
By repeating a polynomial number of times an algorithm working in total polynomial time
and returning the set of solutions we find in the majority of runs, we decrease exponentially the probability of error.
The choice of $\frac{2}{3}$ is hence arbitrary, everything greater than $\frac{1}{2}$ would do.
This property holds for the other probabilistic classes we are going to introduce, but unlike $\TotalPP$
the predicate which defines the enumeration problem needs then to be decidable in polynomial time

Early termination versions of Zippel's algorithm \cite{zippel1990interpolating,kaltofen2000early} solve $enum{Poly}$
in a time polynomial in the number of monomials. If we now use this algorithm on the Determinant which is computable in polynomial time,
 we enumerate its monomials in probabilistic polynomial total time. Thanks to Example \ref{det}, the enumeration 
of the cycle covers of a graph is in $\TotalPP$.

\begin{definition} 
A problem $\enum{A}$ is decidable in incremental polynomial time $\incremP$ (resp. probabilistic polynomial total time $\incremPP$)
if there is a polynomial $Q(x,y)$ and a machine $M$ which solves $\enum{A}$ (resp. with probability greater than $\frac{2}{3}$) and satisfies for all $x$,
$T(x,i+1)-T(x,i) \leq Q(|x|, i)$.
\end{definition}

The classes $\incremP$ and $\incremPP$ can be related to the following search problem, parametrized
by a polynomially balanced predicate $A$.

\begin{trivlist}
  \item[]
    \textsc{AnotherSolution$_{A}$}\\
    \textit{Input:} An instance $x$ of $A$ and a subset $S$ of  $A(x)$\\
    \textit{Sortie:} An element of $A(x)\setminus S$ or a special value if $A(x) = S$
 \end{trivlist}

It has been proved \cite{DBLP:journals/ipl/KavvadiasSS00} that \textsc{AnotherSolution}$_{A} \in \mathrm{FP}$ if and only if $A \in \incremP$. 
We adapt this result to the class $\incremPP$.
If $A$ is a polynomial predicate, the search problem is to return for all $x$
an element of $A(x)$ or a special value if $A(x)$ is empty.
A search problem has a solution in probabilistic polynomial time if there is
a polynomial time algorithm which solves the search problem with probability $\frac{2}{3}$.

\begin{proposition}
\label{eq_inc}
\textsc{AnotherSolution}$_{A}$ has a solution in probabilistic polynomial time if and only if $A \in \incremPP$.
\end{proposition}

Since Zippel's algorithm finds all monomials in its last step, it seems hard to turn it into an incremental algorithm.
On the other hand Algorithm \ref{incremental} whose design has been inspired by Proposition \ref{eq_inc} does the interpolation
with incremental delay. 

\begin{example}
To each graph we associate a polynomial PerfMatch, whose monomials represent the perfect matchings of this graph.
For graphs with a ``Pfaffian'' orientation, such as the planar graphs, this polynomial is related to a Pfaffian and is then efficiently computable. 
Moreover all the coefficients of this graph are positive, therefore we can use Algorithm  \ref{incremental} to interpolate it deterministically with incremental delay. 
We have then proved that the enumeration of perfect matching is in $\incremP$.
\end{example}

\begin{definition}
A problem $\enum{A}$ is decidable in polynomial delay $\delayP$, (resp. probabilistic polynomial delay time $\DelayPP$)
if there is a polynomial $Q(x,y)$ and a machine $M$ which solves $\enum{A}$ (resp. with probability greater than $\frac{2}{3}$) and satisfies for all $x$,
 $T(x,i+1)-T(x,i) \leq Q(|x|)$.  
\end{definition}

\begin{example}[Spanning Hypertrees]
 The notion of a spanning tree of a graph has several interesting generalizations to the case of hypergraphs.
Nevertheless deciding if there is a spanning hypertree is polynomially computable only for the 
notion of Berge acyclicity and $3$-uniform hypergraphs \cite{phd_duris} thanks to an adaptation of the Lov\'asz matching 
algorithm in linear polymatroids \cite{lovasz1980matroid}. 

A polynomial $Z$ is defined for each $3$-uniform hypergraph \cite{masbaum2002new} with coefficients $-1$ or $1$, whose monomials are in bijection 
with the spanning hypertrees of the hypergraph. A new Matrix-Tree theorem \cite{masbaum2002new} shows that $Z$
is the Pfaffian of a matrix, whose coefficients are linear polynomials depending on the hypergraph. 
 Thus $Z$ is efficiently computable by first evaluating a few linear polynomials and then a Pfaffian.
This has been used to give a simple $\mathrm{RP}$ algorithm \cite{caracciolo2008randomized} to decide the existence of a spanning hypertree in 
a $3$-uniform hypergraph.

If we use Algorithm \ref{pdelai} we can enumerate the monomials of $Z$ with probabilistic polynomial delay. The delay is good since the total degree of the monomials
is low and the size of the coefficients is $1$, which helps in the interpolation of the univariate polynomials.
As a conclusion, the problem of enumerating the spanning hypertrees of a $3$-uniform hypergraph is in $\DelayPP$.
\end{example}

\section{Degree $2$ Polynomials}

\subsection{An Incremental Algorithm for Degree 2 Polynomials}

We now give an incremental algorithm for the case of polynomials of degree $d=2$.
It is enough to describe a procedure which finds a monomial of a polynomial $P$,
then Algorithm \ref{incremental} turns it into an incremental algorithm.

First remark that we may use Algorithm \ref{find_monomial} on a polynomial $P$
of arbitrary degree to find a minimal support $L$ in $P$. 
 Since it is minimal, all monomials of $P_{L}$ have $L$ as support and $P_{L}(\vec{X}) = \displaystyle{ \vec{X}^{L}Q(\vec{X})}$ with
 $Q$ a multilinear polynomial. Therefore if we find a monomial of $Q(\vec{X})$
and multiply it by $\vec{X}^{L}$, we have a monomial of $P$.

 We may simulate an oracle call to $Q(\vec{X})$ by a call to the oracle giving $P_{L}$  and a division by 
 the value of $ \vec{X}^{L}$ as long as no $X_{i}$ is chosen to be $0$. 
Remark that the procedure \emph{not\_zero\_improved}($L', L\setminus L',Q,\epsilon$) calls the black box only on 
strictly positive values since  $L = L' \cup (L\setminus L')$. 
It allows us to decide if $Q$ has a monomial whose support contains $L'$.
In Algorithm \ref{deg2} we find a $L'$ such that it is contained in the support of a monomial and is maximal for this property.
Since $Q(\vec{X})$ is multilinear there is only one monomial of support $L'$ and 
we find its coefficient by the procedure \emph{coefficient($Q$,$L'$)}.

\begin{algorithm}[H]
\caption{Finding a monomial of a degree two polynomial}
\DontPrintSemicolon
 \KwData{A polynomial $P_{L} = \displaystyle{ \vec{X}^{L}Q(\vec{X})}$ of degree $2$ with $n$ variables,
  an error bound $\epsilon$}
\KwResult{A monomial of $Q$}
\Begin{
$L' \longleftarrow \varnothing$\;
\For{$i=1$ \KwTo $n$}{
\If{ \emph{not\_zero\_improved}($\varnothing,L' \cup \{i\},Q,\frac{\epsilon}{n}$)}{$L'  \longleftarrow L' \cup \{i\}$}
}
\Return{\emph{coefficient}($Q,L'$)}
}
\label{deg2}
\end{algorithm}

Thanks to Algorithm \ref{deg2} we have a monomial of $Q$ and if we multiply it by  $ \vec{X}^{L}$ it is a monomial of $P$.
We then use it to implement \emph{find\_monomial} in Algorithm \ref{PincremP} and obtain an incremental interpolation 
algorithm for degree $2$ polynomials. 

\subsection{Limit to the Polynomial Delay Approach}

Here we study the problem of deciding if a monomial has coefficient zero in a polynomial.
In the case of a multilinear polynomial the procedure \emph{not\_zero\_improved} solves the problem
in polynomial time but for degree $2$ polynomials we prove it is unpossible unless $\mathrm{RP} = \NP$. 
Therefore there is no generalization of Algorithm \ref{pdelai} to higher degree polynomials, although
a polynomial delay algorithm may exist.

\begin{proposition}
 Assume there is an algorithm which, given a polynomial of degree $2$ and a monomial, can decide
in probabilistic polynomial time if the monomial appears in the polynomial then $\mathrm{RP} = \NP$. 
\end{proposition}

\begin{proof}
Let $G$ be a directed graphs on $n$ vertices, the Laplace matrix $L(G)$ is defined by 
 $L(G)_{i,j} = - X_{i,j}$ when $(i,j) \in E(G)$, $L(G)_{i,i} = \displaystyle{\sum_{(i,j) \in E(G)}X_{i,j}}$  and $0$ otherwise.
The Matrix-Tree theorem is the following equality where $\mathcal{T}_s$ is the set of spanning trees of $G$
whose all edges are oriented away from the vertex $s$ and $L(G)_{s,t}$ is the minor of $L(G)$ where row $s$ and column $t$
have been deleted:
$$ \det(L(G)_{s,t})(-1)^{s+t} = \displaystyle{\sum_{T \in \mathcal{T}_s}\prod_{(i,j) \in T} X_{i,j}}$$
We substitute to $X_{i,j}$ the product of variables $Y_i Z_j$ in the polynomial $\det(L(G)_{s,t})$
which makes it a polynomial in $2n$ variables still computable in polynomial time.
Every monomial represents a spanning tree whose maximum outdegree is the degree of the polynomial.
We assume that every vertex of $G$ has indegree and outdegree less or equal to $2$
therefore $\det(L(G)_{s,t})$ is of degree $2$.

Remark now that a spanning tree, all of whose vertices have outdegree and indegree less or equal to $1$
is an Hamiltonian path. Therefore $G$ has an Hamiltonian path beginning by $s$ and finishing by a vertex $v$
if and only if $\det(L(G)_{s,t})$ contains the monomial $Y_s Z_v \prod_{i \neq s,v} Y_i Z_j$.

There is only  a polynomial number of pairs $(s,v)$, thus if we assume there is a probabilistic polynomial time algorithm to test if a monomial is 
in a degree $2$ polynomial, we can decide in probabilistic polynomial time if $G$ of outdegree and indegree at most $2$ has an Hamiltonian path.
Since this problem is $\NP$ complete \cite{Plesnik79} we have $\mathrm{RP} = \NP$.
\end{proof}

\section{Conclusion}

Let us compare our method to three classical interpolation algorithms,
which unlike our method can interpolate polynomials of any degree.
Once restricted to multilinear polynomials, Algorithm \ref{pdelai} is really efficient
compared to the algorithm of Klivans and Spielman (KS), which is the only known method with a bound on the delay.
Note also that Algorithm \ref{incremental}, which is not presented in the next table, needs  $n^{2}$ calls 
to the black box to guess one monomial, whereas KS needs $(nD)^{6}$ calls 
and then the same method is used to recover the whole polynomial from this procedure.

In the table $T$ is a bound on $t$ the number of monomials that Ben-Or Tiwari and Zippel algorithms need to do the interpolation.
In the row labeled Enumeration is written the kind of enumeration algorithm the interpolation method 
gives when the polynomial is polynomially computable.

\begin{center}
\begin{tabular}[c]{|l||l|l|l|l|}
 \hline
  & Ben-Or Tiwari \cite{ben1988deterministic} &  Zippel \cite{zippel1990interpolating} & KS \cite{klivans2001randomness} & Algorithm \ref{pdelai}\\
\hline
Algorithm type& Deterministic & Probabilistic & Probabilistic & Probabilistic  \\
\hline
Number of calls & $2T$ &  $tnD$ & $t(nD)^{6}$ & $tnD(n + \log(\frac{1}{\epsilon}))$   \\
\hline
Total time & Quadratic in $T$ &  Quadratic in $t$\, &  Quadratic in $t$\, & Linear in $t$   \\
\hline
Enumeration & Exponential & $\TotalPP$ & $\incremPP$ & $\DelayPP$\\ 
\hline
Size of points & $T\log(n)$ & $\log(\frac{nT^{2}}{\epsilon})$& $\log(\frac{nD}{\epsilon})$ & $\log(D)$\\
\hline
\end{tabular}
\end{center}

\medskip

\emph{Acknowledgements}   
Thanks to Herv\'e Fournier ``l'astucieux'', Guillaume Malod, Sylvain Perifel and Arnaud Durand for their helpful comments
about this article.

\bibliographystyle{splncs}
\bibliography{ma_biblio.bib}

\section*{Appendix}

Here we give most of the proofs which are omitted in the article and the alternate 
implementation of \emph{not\_zero\_improved} with only one oracle call.

\bigskip

\noindent
\textbf{Proof of Theorem \ref{PincremP} :}

\noindent
\textbf{Correction :} \\
We analyze this algorithm under the assumption that the procedures \emph{not\_zero} and \emph{find\_monomial}
do not make mistakes.

We have the following invariant of the while loop :
$Q$ is made from a subset of the monomials of $P$.
It is true at the beginning because $Q$ is zero.
Assume that $Q$ satisfies this property at a certain point of the while loop,
since we know that \emph{not\_zero}(\emph{subtract}($P$,$Q$)),
 $P-Q$ is non zero and is then a non empty subset of the monomials
 of $P$. The outcome of \emph{find\_monomial}(\emph{subtract}($P$,$Q$)) 
is thus a monomial of $P$ which is not in $Q$, therefore $Q$ plus this monomial 
still satisfies the invariant.
Remark that we have also proved that the number of monomials of $Q$ is increasing by 
one at each step of the while loop.
The algorithm must then terminate after $t$ steps and when it does \emph{not\_zero}(\emph{subtract}($P$,$Q$))
gives a negative answer meaning that $Q = P$.

\noindent
\textbf{Probability of error :} \\
The probability of failure is bounded by the sum of the probabilities of error coming from \emph{not\_zero} and \emph{find\_monomial}.
We both call these procedures $t$ times with an error bounded by $\frac{\epsilon}{2^{n+1}}$. Since
$2t \leq 2^{n+1}$, the total probability  of error is bounded by $\epsilon$.

\noindent
 \textbf{Complexity :} \\
The procedure \emph{not\_zero} is called $t$ times and uses the oracle $n + \log(\frac{1}{\epsilon})$ times, whereas
\emph{find\_monomial} is called $t$ times but uses $n(n +\log(\frac{1}{\epsilon}))$ oracle calls, which adds up to $t(n+1)(n +\log(\frac{1}{\epsilon}))$ calls to the oracle.
In both cases the evaluation points are of size $O(D)$.

The delay between two solutions is bounded by the evaluation of  \emph{find\_monomial}, which 
is dominated by the execution of \emph{subtract}$(P,Q)$ at each oracle call on points of size
$D$. The algorithm calls \emph{subtract}$(P,Q)$ $n(n +\log(\frac{1}{\epsilon}))$ times and each of these calls needs $O( iD\max(C,D))$, which gives a delay of $O( iD\max(C,D)n(n +\log(\frac{1}{\epsilon})))$. 
\qed

\bigskip

\noindent
\textbf{Alternate method to implement \emph{not\_zero\_improved} :}

We want to decide if $P_{1}(\vec{X})$ is the zero polynomial in  $P_{L_{1}\cup L_{2}} = \displaystyle{ \vec{X}^{L_{2}} P_{1}(\vec{X}) + P_{2}(\vec{X})}$.
We let $\alpha$ be the integer $2^{2 (n + C + D \log(\frac{2D}{\epsilon}))}$ and $l$ the cardinal of $L_{2}$.
We do a call to the oracle on the values $(x_{i})_{i \in [n]}$:

$$ \left\{ \begin{array}{l l}
     x_{i} \text{ is randomly chosen in }  [\frac{2D}{\epsilon}] &\text{ if } i \in L_{1} \\
      x_{i} = \alpha & \text{ if } i \in L_{2}\\
      x_{i} = 0   & \text{otherwise}
    \end{array} \right.
 $$
 
The value of a variable which is not in $L_{2}$ is bounded by $\frac{2D}{\epsilon}$, therefore a monomial 
of $P_{2}$ (which contains at most $l-1$ variables of $L_{2}$) has its contribution to $P(x_{1},\dots,x_{n})$ bounded by $2^{C}(\frac{2D}{\epsilon})^{D}\alpha^{l-1} $.
Hence the total contribution of $P_{2}$ is bounded in absolute value by $2^{n + C + D \log(\frac{2D}{\epsilon})} \alpha^{l-1}$ which is equal to
$\alpha^{l - \frac{1}{2}}$. If $P_{1}(x_{1},\dots,x_{n})$ is zero, this also bounds the absolute value of $P(x_{1},\dots,x_{n})$.

Assume now that $P_{1}(x_{1},\dots,x_{n})$ is not zero, since $\vec{x}^{L_{2}} $ is equal to $\alpha^{l}$, the absolute value 
of $\displaystyle{ \vec{x}^{L_{2}} P_{1}(x_{1},\dots,x_{n})}$ has $\alpha^{l}$ for lower bound.
By the triangle inequality 
$$ \begin{array}{l}
|P(x_{1},\dots,x_{n})| > \left| |\displaystyle{\vec{x}^{L_{2}} P_{1}(x_{1},\dots,x_{n})}|-|P_{2}(x_{1},\dots,x_{n})| \right| \\
|P(x_{1},\dots,x_{n})| > \alpha^{l} - \alpha^{l-\frac{1}{2}} > \alpha^{l-\frac{1}{2}}
\end{array} $$

We can then decide if $P_{1}(x_{1},\dots,x_{n})$ is zero by comparison of $P(x_{1},\dots,x_{n})$ to $\alpha^{l-\frac{1}{2}}$.
Remark that $P_{1}(x_{1},\dots,x_{n})$ may be zero even if $P_{1}$ is not zero. Nonetheless $P_{1}$ only depends on variables 
which are in $L_{1}$ and are thus randomly taken in $[\frac{2D}{\epsilon}]$. By Lemma \ref{proba}, the probability that the polynomial $P_{1}$ 
is not zero although $P_{1}(x_{1},\dots,x_{n})$ has value zero is bounded by $\epsilon$. We have then designed an algorithm which decides with
probability $1 - \epsilon$ if $P$ has a monomial whose support contains $L_{2}$ and is contained in $L_{1} \cup L_{2}$.

Remark that this implementation of \emph{not\_zero\_improved} needs only one oracle call but requires big evaluation points and
to know $C$ in advance. To implement \emph{coefficient} just do the same oracle call and an integer division of the value by $\alpha^{l}$
to get the coefficient.

\bigskip

\noindent 
\textbf{Proof of Theorem \ref{PdelayP} :}

The procedure \emph{not\_zero\_improved} does one interpolation on a degree $l$ polynomial where $l$ is bounded
by $D$. We can bound the value of the polynomial on points of value less than $D$ by $2^{n}2^{C}D^{D}$,
where $C$ is a bound on the size of the coefficients of the polynomial. Since we have assumed that $C = O(n)$
and that $D < n$ because the polynomial is multilinear, the logarithm of the values taken by the polynomial for the interpolation
is bounded by $n\log(n)$. The univariate interpolation then needs a time $O(D^{2}n\log(n))$ and 
$D$ oracle calls on points of size $\log(D)$. 

The procedure \emph{not\_zero\_improved} is called in Algorithm \ref{pdelai} with an error parameter $\frac{\epsilon}{n2^{n}}$,
it therefore repeats  the previously described interpolation $O(n + \log(\frac{1}{\epsilon}))$ times.
Each call to \emph{not\_zero\_improved} needs a time $O(D^{2}n\log(n))(n + \log(\frac{1}{\epsilon}))$ and 
$D(n + \log(\frac{1}{\epsilon}))$ oracle calls.

Between the visit of two leaves, we call the procedure \emph{not\_zero\_improved} at most $n$ times and once the procedure \emph{coefficient}
which has a similar complexity. Hence the delay is bounded in time by $O(D^{2}n^{2}\log(n))(n + \log(\frac{1}{\epsilon}))$ and by 
$nD(n + \log(\frac{1}{\epsilon}))$ oracle calls on points of size $\log(2D)$.

Finally since we call the procedures \emph{not\_zero\_improved} and \emph{coefficient} less than $nt$ times during the algorithm,
the error is bounded by $nt\frac{\epsilon}{n2^{n}} < \epsilon$.
\qed

\bigskip

\noindent 
We describe here the way to a improve the error bound for  $\incremPP$ algorithms,
but it would work equally well on $\DelayPP$ ones. Note that in both cases we need an exponential 
space and there is a slight overhead.
\begin{proposition}
\label{amplification}
 If a problem $A$ is in $\incremPP$ then there is a polynomial $Q$ and a machine $M$ which for all $\epsilon$ computes the solution of $A$ with probability $1-\epsilon$ and satisfies for all $x$,
$T(x,i+1)-T(x,i) \leq Q(|x|, i )\log(\frac{1}{\epsilon})$.  
\end{proposition}

\begin{proof}
 Since $A$ is in $\incremPP$, there is a machine $M$ which computes the solution of $A$ with probability $\frac{2}{3}$ and a delay bounded by $Q(|x|, i )$.
 Since $A(x,y)$ may be tested in polynomial time, we can assume that every output of $M$ is a correct solution, by checking $A(x,y)$ before outputting $y$
 and stopping if not $A(x,y)$.
We now simulate $k$ runs in parallel of the machine $M$ on input $x$.
Each time we should output a solution, we add it to a set of solution (with no repetition).
Assume we have already outputted $i$ solutions, we let the $k$ runs be simulated for another $Q(|x|, i )$ steps each before outputting
a new solution of the set of found solutions and stop if it is empty.
This algorithm clearly works in incremental polynomial time and if one of the run finds all solutions, it also finds all solutions.
Then the probability of finding all solutions is more than $1 - \frac{1}{3}^{k}$. If we set $k = \frac{\log(\frac{1}{\epsilon})}{\log(3)}$,
we have a probability of $1-\epsilon$, which achieves the proof.
\end{proof}

\bigskip

\noindent
\textbf{Proof of Proposition \ref{eq_inc} :}

 Assume  \textsc{AnotherSolution}$_{A}$ is computable in probabilistic polynomial time, we want to enumerate the solution of the enumeration problem
 $A$ on the input $x$. We know a bound on the number of solutions of $A$, that we call $B$. We assume that the algorithm which decides  \textsc{AnotherSolution}$_{A}$ 
 has a probability of error of $\frac{1}{3B}$. That is achievable by repeating at most $\log(B)$ times the original algorithm, therefore the running time is still polynomial. 
 We apply  this algorithm to $x$ and the empty set, we add the found solution to the set of solutions
 and we go on like this until we have found all solutions. The delay between the $i^{\text{th}}$ and the $i+1^{\text{th}}$ solution is bounded by the execution 
 of the algorithm \textsc{AnotherSolution} which is polynomial in $|x|$ and $i$ the size of the set of already found solutions. Moreover the probability of error is bounded
 by $\frac{1}{3} = B \times \frac{1}{3B}$. This proves that $A$ is in $\incremPP$. 
 
 Conversely if $A \in \incremPP$, on an instance  ($x$,$S$) of \textsc{AnotherSolution}$_{A}$ we want to find a solution which is not in $S$.
We enumerate $|S|+1$ solutions by the $\incremPP$ algorithm, in time polynomial in $|S|$ and $|x|$. If one of these solutions is not in $S$,
it is the output of the algorithm. If $S$ is the set of all solutions, the enumeration will end in time polynomial in $S$ and $x$, which allow us
to output the value meaning there is no other solutions.
\qed

\end{document}